\documentclass[reqno,11pt,a4]{amsart}
\usepackage{amsmath,amssymb,tabularx,setspace,color,multirow,graphics,graphicx}

\usepackage{adjustbox}

\newtheorem{definition}{Definition}
\newtheorem{remark}{Remark}

\newtheorem{thm}{Theorem}[section]

\newcommand{\be}{\begin{equation}}
\newcommand{\ee}{\end{equation}}
\newcommand{\bee}{\begin{eqnarray}}
\newcommand{\eee}{\end{eqnarray}}
\newcommand{\bees}{\begin{eqnarray*}}
	\newcommand{\eees}{\end{eqnarray*}}

\newcommand{\bc}{\begin{center}}
	\newcommand{\ec}{\end{center}}
\newcommand{\bt}{\begin{table}[!h]}
	\newcommand{\et}{\end{table}}

\makeatletter
\renewcommand\section{\@startsection {section}{1}{\z@}%
                                   {-3.5ex \@plus -1ex \@minus -.2ex}%
                                   {2.3ex \@plus.2ex}%
                                   {\normalfont\large\bfseries}}
\makeatother
\newtheorem{Definition}{Definition}
\newtheorem{Corollary}{Corollary}

\newtheorem{Theorem}{Theorem}

\newtheorem{Example}{Example}

\begin{document}
\title[Entropy generating function for past lifetime]{Entropy generating function for past lifetime and its properties}
\author[]{S\lowercase{mitha} S$^{\lowercase{a}}$ \lowercase{and} S\lowercase{udheesh} K. K\lowercase{attumannil}$^{\lowercase{b}}$\\
 $^{\lowercase{a}}$ K E C\lowercase{ollege} M\lowercase{annanam,} K\lowercase{erala}, I\lowercase{ndia},\\
 $^{\lowercase{b}}$I\lowercase{ndian} S\lowercase{tatistical} I\lowercase{nstitute},
  C\lowercase{hennai}, I\lowercase{ndia.}}
\thanks{{$^{\dag}$}{Corresponding E-mail: \tt skkattu@isichennai.res.in }}

\doublespace
\begin{abstract}
	 The past entropy is considered as an uncertainty measure for the past lifetime distribution. Generating function approach to entropy become popular in recent time as it generate several well-known entropy measures.  In this paper, we introduce the past entropy-generating function. We study certain properties of this measure. It is shown that the past entropy-generating function uniquely determines the distribution. Further, we present characterizations for some lifetime models using the relationship between reliability concepts and the past entropy-generating function.
\\\textit{keywords:}
	Past entropy generating function, Lifetime models, Characterizations.
\end{abstract}
\maketitle
\section{Introduction}
\label{sec:intro}
In many practical situations, measuring the uncertainty associated with a random variable is very important. In this connection, the notion of entropy was originally developed by Shannon (1948).  It is of fundamental importance in different areas such as Physics, Communication theory, Probability, Economics, etc. Let $X$ be an absolutely continuous non-negative random variable having cumulative distribution function $F(x)=P(X\leq x)$, probability density function $f(x)$ and survival function $\bar{F}(x)=P(X>x)$.  The Shannon's entropy is defined as
\begin{equation}\label{ent}
H(F)=-\int\limits_{0}^{\infty} {f(x)\log f(x) dx},
\end{equation}where `log' denotes the natural logarithm.
 Several extensions of (\ref{ent}) are available in literature each one being suitable for some specific situation. Some of these measures are; cumulative residual entropy (Rao et al., 2004), cumulative past entropy (Di Crescenzo and Longobardi, 2009), and the corresponding weighted and dynamic measures (Chakraborty and Pradhan, 2023).  Sudheesh et al. (2022) proposed general measures of cumulative (residual) entropy showing that several measures discussed in the literature are special cases of their proposed measures of entropy.

A convenient tool to generate moments of a probability distribution is its moment-generating function (m.g.f). It is well known that the successive derivatives of the m.g.f at point zero give the successive moments of a probability distribution, provided these moments exist. Golomb (1966)  introduced entropy generating function of a probability distribution and is given by
\begin{equation}\label{gf}
B_s(F)=\int_0^{\infty}f^s(x)dx,\quad s\geq1.
\end{equation}
The first order derivative of $B_s(F)$ evaluated at value one yields negative Shannon entropy or negative entropy of a distribution.  Guiasu and  Reischer (1985)  proposed a  relative information-generating function and they observed that its derivatives provide well-known statistical measures like the Kullback-Leibler divergence between two probability distributions. Arora and Nath (1972) proposed an inaccuracy-generating function for two probability distributions and studied its relationship with the entropy-generating function proposed by Golomb. The first derivative of this function at point one gives the inaccuracy measure by Kerridge (1961).

Recently, information-generating functions have received much attention. Clark (2019) introduced an information-generating function for point processes for determining statistics related to entropy and relative entropy.  In a series of papers, Kharazmi and Balakrishnan (2021a, 2021b, 2021c) introduced cumulative residual information generating and relative cumulative residual information generating measures, cumulative residual Fisher information and relative cumulative residual Fisher information measures, and studied their properties. In particular, Kharazmi and Balakrishnan (2021a)  proposed two new divergence measures and showed that some of the well-known information divergences such as Jensen-Shannon, Jensen-extropy and Jensen-Taneja divergence measures are all special cases of it.
For the information-generating functions associated with maximum and minimum ranked sets and record values and their properties, see Zamani et al. (2022) and Kharazmi et al. (2021).
Saha and Kayal (2023) proposed general weighted information and relative information-generating functions and studied their properties. Recently, Smitha et al. (2023) proposed a dynamic residual entropy-generating function and studied its properties. Motivated these recent development, as past life plays an important role in lifetime data analysis, we consider  entropy generating function associated with past lifetime and study its properties.

The rest of the paper is organised as follows.  In section 2, we give the definition of past entropy generating function and propose the relationship with reversed hazard rate. Various properties of this measure are also discussed. In Section 3,  we give some characterization results based on the functional relationship between past entropy-generating functions and reliability concepts. Finally, in Section 4, we give some concluding remarks.
\vspace{-0.2in}
\section{Past Entropy Generating Function and its Properties}\vspace{-0.2in}
In lifetime analysis, on several occasions, one has information about the current age of the unit under consideration. In such cases,  we define a random variable $X_t^*=X-t|X\geq t$ to study the properties of the unit.  The mean $E(X_t)$ is known as the mean residual life and it has many applications. Measuring the uncertainty about the remaining lifetime of the unit, Shannon's entropy associated with $X_t$ was developed by Ebrahimi and Pellery (1995) and  has the form
\begin{equation*}\label{2}
H(F:t)=-\int\limits_t^{\infty} {\frac{f(x)}{\bar{F}(t)}\log \frac{f{(x)}}{\bar{F}(t)} dx}.
\end{equation*}

Apart from the residual life of a unit, past life also has important applications in lifetime analysis.  In this case,  we are interested in studying the behavior of the random variable $X_t=t-X|X\leq t$.
Di Crescenzo and Longobardi (2004) have shown that in many situations, uncertainty is not necessarily related to the future but can also refer to the past. By considering such situations, they proposed the past entropy defined over $(0,t)$ and is given by
\begin{equation}\label{5}
\bar{H}(F;t)=-\int\limits_0^{t} {\frac{f(x)}{F(t)}\log \frac{f{(x)}}{F(t)} dx}.
\end{equation}
We observed that $\bar{H}(F;t)$ can be viewed as the entropy of inactivity time $t-X|X\leq t.$ This measure plays an important role in the context of information theory, forensic science, and other related fields.

The past entropy function can also be expressed in terms of the reversed hazard rate,
$\lambda(x)=\frac{f(x)}{F(x)}$, through the relationship
\begin{equation}\label{6}
\bar{H}(F;t)=1-\frac{1}{F(t)}\int_0^tf(x)\log \lambda(x)dx.
\end{equation}
For more works on entropy related to past lifetime,  we refer to Di Crescenzo and Longobardi (2004,2006), Nanda and Paul (2006), Kundu et al. (2010), Mini (2016), Nair et al. (2021), Sudheesh et al. (2022) and Chakraborty and Pradhan (2023a).  Next, we define the entropy-generating function in the past lifetime.
\begin{Definition}\vspace{-0.25in}
 Let $X$ be an absolutely continuous non-negative random variable with cumulative distribution function $F(x)=P(X\leq x)$ and probability density function $f(x)$.  Let $X_t=t-X|X\leq t$ be the past life associated with $X$. Then the past entropy generating function is given by
 \begin{equation}\label{pegf}
\bar{B_s}(F;t)=\int_0^t\left(\frac{f(x)}{F(t)}\right)^s dx,\quad s\geq1.
\end{equation}
\end{Definition}

In the following example, we show that two distributions can possess the same entropy generating function, but different past entropy generating function.
\begin{Example}
 Consider two components in a system having random lifetimes $X$ and $Y$  with distribution functions $F$ and $G$ and density functions
 $$
 f_{X}(x)=\frac{1}{2} ,\,0\leq x\leq 2
 $$
 and
 $$
 f_{Y}(y)=\frac{1}{2} ,\,2\leq y\leq4.
 $$
 Here
$$
\bar{B_s}(X)=2^{1-s}
$$
and
$$
\bar{B_s}(Y)=2^{1-s},
$$
respectively. We can observe that $\bar{B_s}(X)= \bar{B_s}(Y)=2^{1-s}$.

In this case, the past entropy-generating function is obtained as
$$
\bar{B_s}(X;t)=\frac{2^st^{1-s}}{s+1}
$$
and
$$
\bar{B_s}(Y;t)=\frac{2^s[1-(1-t)^{s+1}]}{(s+1)(2t-t^2)^{s}}.
$$
This result implies that the expected uncertainty regarding the predictability of the outcomes of $X$
and $Y$ in terms of entropy generating functions are identical for $F$ and $G$. However,  the past entropy-generating functions of $X$  and  $Y$  are different.
\end{Example}
Next, we study the properties of the past entropy-generating function. \\
\noindent\textbf{Property 1:} The first derivative of $\bar{B_s}(F;t)$ at $s=1$ with negative sign is the past entropy measure proposed by DiCrescenzo and Longobardi (2004).
\\\textbf{Property 2:} For any absolutely continuous random variable $X,$ define $Z=aX+b,$ where $a>0,$ and $b\geq0,$ are constants. Then, for $t>b$
$$\bar{B_s}(Z;t)=a^{1-s}.\bar{B}_s\left(X;\frac{t-b}{a}\right).$$
\begin{proof}
	\begin{eqnarray*}
		\bar{B_s}(Z;t)&=&\int_0^t\left(\frac{f_Z(x)}{F_Z(t)}\right)^s dx
		\\&=&\frac{1}{\left[F\left(\frac{t-b}{a}\right)\right]^s}\int_0^t\left(\frac{1}{a}f\left(\frac{x-b}{a}\right)\right)^s dx
		\\&=&a^{1-s}\int_0^{\frac{t-b}{a}}\left(\frac{f(u)}{F\left(\frac{t-b}{a}\right)}\right)^s du
		\\&=&a^{1-s}\bar{B}_s\left(X;\frac{t-b}{a}\right).
	\end{eqnarray*}
	Hence the proof.
\end{proof}
\newpage
\noindent From Property 2, we have the following identity.
	\begin{enumerate}
		\item If $b=0,$ $\bar{B}_s(aX;t)=a^{1-s}\bar{B_s}(X;\frac{t}{a})$
		\item If $a=1,$ $\bar{B}_s(X+b;t)=\bar{B_s}(X;t-b)$
	\end{enumerate}
\noindent \textbf{Property 3:} The relationship between  $\bar{B_s}(F;t)$ and reversed hazard rate is given by
\begin{equation}\label{8}
\bar{B_s}^{\prime}(F;t)=\left(\lambda(t)\right)^s-s.\bar{B_s}(F;t)\lambda(t),
\end{equation}where $\bar{B_s}^{\prime}(F;t)$ is the derivative of $\bar{B_s}(F;t)$ with respect to  $t$.

\section{Characterization Results}
In this section, we establish several characterization results associated with past entropy-generating functions.  The following result shows that the past entropy-generating function uniquely determines the underlying distribution function.
\begin{thm}
	Let $F(x)$ be an absolutely continuous distribution function and assume that $\bar{B_s}(F;t)$ is increasing in $t.$ Then $\bar{B_s}(F;t)$ uniquely determines $F(t).$
\end{thm}
\begin{proof}
	Suppose that $F(x)$ and $G(x)$ are distribution functions such that
	\begin{equation}\label{9}
	\bar{B_s}(F;t)=\bar{B_s}(G;t),~\forall~t\geq 0.
	\end{equation}
	That is,
	\begin{equation}\label{10}
	\frac{1}{\left(F(t)\right)^s}\int_0^tf^s(x)dx=	\frac{1}{\left(G(t)\right)^s}\int_0^tg^s(x)dx,~\forall~s\geq1,
	\end{equation}
	where $f(x)$ and $g(x)$ are probability generating functions corresponding to $F(x)$ and $G(x)$, respectively. Differentiating \eqref{10} with respect to  $t$, we obtain
	\begin{eqnarray*}
	    (F(t))^{-s}(f(t))^s&+&(-s)(F(t))^{-s-1}f(t)\int_0^tf^s(x)dx\\&=&(G(t))^{-s}(g(t))^s+(-s)(G(t))^{-s-1}g(t)\int_0^tg^s(x)dx.
	\end{eqnarray*}This implies
	$$(\lambda_1(t))^s-s\lambda_1(t)\bar{B_s}(F;t)=(\lambda_2(t))^s-s\lambda_2(t)\bar{B_s}(G;t),$$
	where $\lambda_1(t)$ and $\lambda_2(t)$ are the reversed hazard rates of $f(x)$ and $g(x)$ respectively. Note that the reversed hazard rate determines the distribution, uniquely. Hence to prove $F(t)=G(t),$ it is enough to prove $\lambda_1(t)=\lambda_2(t),~\forall~t\geq0.$

	Suppose we have the inequality $$\lambda_1(t)>\lambda_2(t),\quad \lambda_i(t)\not=0,~i=1,2.$$ From \eqref{10} we obtain
	$$\lambda_1(t)\left[(\lambda_1(t))^{s-1}-s\bar{B_s}(F;t)\right]=\lambda_2(t)\left[(\lambda_2(t))^{s-1}-s\bar{B_s}(G;t)\right].$$
	That is
	\begin{equation}\label{11}
	\frac{\lambda_1(t)}{\lambda_2(t)}=\frac{(\lambda_2(t))^{s-1}-s\bar{B_s}(G;t)}{(\lambda_1(t))^{s-1}-s\bar{B_s}(F;t)}
	\end{equation}
	$$\frac{\lambda_1(t)}{\lambda_2(t)}>1\implies(\lambda_2(t))^{s-1}-s\bar{B_s}(G;t)>
	(\lambda_1(t))^{s-1}-s\bar{B_s}(F;t)$$
	Using \eqref{9}, the above equation gives
	$$(\lambda_2(t))^{s-1}>(\lambda_1(t))^{s-1}\implies\lambda_2(t)>\lambda_1(t).$$
	This is a contradiction. Similarly, we can show that the inequality $\lambda_1(t)<\lambda_2(t)$ also leads to a contradiction. Therefore, we have  $\lambda_1(t)=\lambda_2(t).$
\end{proof}	
\noindent Now we give the characterization results for uniform and power function distributions using the functional relationship between the past entropy generating function and reversed hazard rate.
\begin{thm}
	Let  $X$ be a continuous random variable with absolutely continuous distribution function $F(x)$ and reversed hazard rate $\lambda(t).$ Then the past entropy generating function $\bar{B_s}(F;t)$ has the form
	\begin{equation}\label{12}
	\bar{B_s}(F;t)=(\lambda(t))^{s-1}~;~s>1
	\end{equation}
	if and only if $X$ follows uniform distribution with distribution function
	\begin{equation}\label{13}
	F(x)=\frac{x-a}{b-a}~;~a\leq x\leq b.\end{equation}
\end{thm}
\begin{proof}
	Suppose $X\sim U[a,b],$  by direct computation  of $\bar{B_s}(F;t)$  we have the necessary part. Conversely, assume \eqref{12} holds. On differentiating both sides of the equation (\ref{12}) with respect to $t$, we obtain
	\begin{equation}\label{14}
	\bar{B_s}^{\prime}(F;t)=(s-1)(\lambda(t))^{s-2}\lambda^{\prime}(t).
	\end{equation}
	Using the relationship between  $\bar{B_s}(F;t)$ and reversed hazard rate given in \eqref{8}, (\ref{14}) becomes
 \begin{equation}\label{121}
     (\lambda(t))^s-s\bar{B_s}(F;t)\lambda(t)=(s-1)(\lambda(t))^{s-2}\lambda^{\prime}(t).
 \end{equation}
 Substituting (\ref{12}) in \eqref{121}, we obtain
	\begin{eqnarray*}
		(\lambda(t))^s-s(\lambda(t))^s=(s-1)(\lambda(t))^{s-2}\lambda^{\prime}(t).
	\end{eqnarray*}Or
 \begin{equation*}
     (\lambda(t))^s(1-s)=(s-1)(\lambda(t))^{s-2}\lambda^{\prime}(t).
 \end{equation*}
	That is, $$\frac{\lambda^{\prime}(t)}{(\lambda(t))^2}=-1,$$
	which can be written as $$\frac{d}{dt}\left(\frac{1}{\lambda(t)}\right)=1.$$
	Integrating both sides of the above equation in the interval $(0,t)$, we obtain the $\lambda(t)$ as the reversed hazard rate of $U(a,b)$ distribution. Since the reversed hazard rate determines the distribution uniquely, we have the model in (\ref{13}).
 \end{proof}	

\noindent Next, we characterize power distribution by using the functional relationship between past entropy generating and reversed hazard rate. The result is stated as follows.
\begin{thm}
	Let $X$ be a non-negative random variable admitting an absolutely continuous distribution function $F(x)$.  The relationship
	\begin{equation}\label{15}
	\bar{B_s}(F;t)=k(\lambda(t))^{s-1}~;~s>1,\,k \in R,
	\end{equation}
	 holds for all real $t\geq0,$ if and only if $X$ has power distribution with distribution function
	$$F(x)=x^c~;~c>0,~0<x<1.$$
\end{thm}
\begin{proof}
	By direct computation we get \eqref{15}. Conversely assume that \eqref{15} holds. Differentiating \eqref{15} and using the relationship between  $\bar{B_s}(F;t)$ and reversed hazard rate given in \eqref{8}, we have
	\begin{eqnarray*}
		(\lambda(t))^{s-1}-sk(\lambda(t))^{s-1}&=&k(s-1)(\lambda(t))^{s-3}\lambda^{\prime}(t)
		\\(1-sk)(\lambda(t))^{s-1}&=&k(s-1)(\lambda(t))^{s-3}\lambda^{\prime}(t)
		\\-\frac{\lambda^{\prime}(t)}{(\lambda(t))^2}&=&\frac{sk-1}{k(s-1)}.
	\end{eqnarray*}The above equation can be written as
 \begin{equation*}
     \frac{d}{dt}\frac{1}{\lambda(t)}=\frac{sk-1}{k(s-1)}.
 \end{equation*}
	Solving the above differential equation we obtain $$\lambda(t)=\frac{c}{t},$$ where $c=\frac{k(s-1)}{sk-1}$, which is the reversed hazard rate of the power distribution.
	\end{proof}
 \begin{thm}
     A non-negative random variable $X$ has an exponential distribution with mean $\mu$ if and only if we have the following relationship holds:
     \begin{equation}\label{exp}
	\bar{B_s}(F;t)=\mu(\lambda(t))^s(\frac{e^{st/\mu}-1}{s}),
	\end{equation}
 \end{thm}
 \begin{proof}
    Let $X$ follows an exponential distribution with mean $\mu$. BY direct calculation we obtain
    $$\bar{B_s}(F;t)=\frac{\mu}{s}\left(\frac{1}{\mu^s(1-e^{-t/\mu})^s}-\left(\frac{e^{-t/\mu}}{\mu(1-e^{-t/\mu})}\right)^s\right).$$In term of the reversed hazard rate, this reduces to the expression given in (\ref{exp}).\\
    Now, suppose (\ref{exp}) holds. Using the definition of $\bar{B_s}(F;t)$ and $\lambda(t)$, we obtain the identity
    $$\int_{0}^{t}f^s(x)dx=\frac{\mu}{s}f^s(t)(e^{st/\mu}-1).$$
    Differentiating both sides with respect to $t$ and simplifying we  obtain
    $$\frac{f'(t)}{f(t)}=\frac{-1}{\mu}.$$ Solving the above differential equation we obtain $f(x)=\frac{1}{\mu}e^{-t/\mu}$, which proves the result.
 \end{proof}
	Next, we consider the case where the random variable $X$ has support $(-\infty,\infty).$  Kundu et al. (2010)  defined a past entropy measure for the support $(a,b),$ $-\infty\leq a<b\leq \infty,$  and studied its property. Let $X$ be a real valued absolutely continuous  random variable having support $(-\infty,\infty)$, then   the past entropy generating function is defined as
	\begin{equation}\label{16}
	\bar{B_s}(F;t)=\int\limits_{-\infty}^t\left(\frac{f(x)}{F(t)}\right)^sdx,~s\ge 1.
	\end{equation}
	For the general support $(a,b),$ where $a=\inf\{x:F(x)>0\}$ and $b=\sup\{x:F(x)<1\},$ $-\infty\leq a<b\leq\infty,$ the past entropy generating function is defined as
	\begin{equation}\label{17}
	\bar{B_s}(F;t)=\int\limits_{a}^t\left(\frac{f(x)}{F(t)}\right)^sdx,~s\ge 1.
	\end{equation}
	If $a=0,$ then \eqref{17} is identical to the past entropy generating function defined in \eqref{pegf}.

The next theorem shows the constancy of the past entropy-generating function.
\begin{thm}
  If $X$ is a random variable admitting an absolutely continuous  distribution function $F(x)$, then the past entropy generating function is independent of  $t$ if and only if the distribution $X$ is specified by
\begin{equation}\label{18a}
F(x)=\exp{[a(x-b)]} ;-\infty \leq x \leq b.
 \end{equation}
\end{thm}
\begin{proof}
Let $\bar{B_s}(F;t)=k $, where $k$ is a positive constant. Then
$$\bar{B_s}^{\prime}(F;t)=0.
$$
Upon using equation (\ref{8}), we obtain
$$
({\lambda(t)})^s- sk \lambda(t) =0.
$$Hence, $\lambda(t)=({s k})^{\frac{1}{s-1}}$,  a constant. By Theorem 2.3 of Nair et al. (2019), the reversed hazard rate is independent of $t$, if the distribution function  of $X$ is specified by  $$F(x)=\exp{[a(x-b)]};-\infty \leq x \leq b.$$

Conversely, assume $X$ has a distribution function specified by (\ref{18a}). Then by direct computation,  we have
$\bar{B_s}(F;t)= \frac{a^s}{as},$
which is independent  of  $t$. This proves the theorem.
\end{proof}
\noindent Next, we give a characterization theorem for generalized power function distribution.
\begin{thm}
	Let $X$ be a continuous random variable with absolutely continuous distribution function $F(x)$ and mean inactivity time $m(t)=E(t-X|X\le t).$ Then $X$ follows generalized power function distribution specified by \begin{equation}\label{18}
	F(x)=\left(\frac{cx+d}{cb+d}\right)^{\frac{1-c}{c}}~;~-\frac{d}{c}<t<b~;~0<c<1
	\end{equation}
	if and only if
	\begin{equation}\label{19}
	\bar{B_s}(F;t)=k(m(t))^{1-s}.
	\end{equation}
\end{thm}
\begin{proof}
	Suppose $X$ follows generalized power function distribution specified in \eqref{18}. Then by direct computation we obtain
	$$\bar{B_s}(F;t)=\frac{\left(c(1-c)\right)^s}{(1-2c)^{s+1}}(ct+d)^{1-s}.$$
	That is,
	$$\bar{B_s}(F;t)=k(m(t))^{1-s}.$$
	Conversely, assume that \eqref{19} holds.  On differentiating both sides of the above equation and using the relationship given in   \eqref{8}, we obtain
	\begin{equation}\label{20}
	(\lambda(t))^s-ks(m(t))^{1-s}\lambda(t)=k(1-s)(m(t))^{-s}m^{\prime}(t).
	\end{equation}
	Multiplying by $(m(t))^s$  through out the equation \eqref{20}  and upon using the relationship  $\frac{1-m^{\prime}(t)}{m(t)}=\lambda(t)$, we obtain
	\begin{equation}\label{21}
	(1-m^{\prime}(t))^s-ks(1-m^{\prime}(t))=k(1-s)m^{\prime}(t).
	\end{equation}
	On simplifying the above equation, we obtain
	\begin{eqnarray*}
		m^{\prime}(t)=1-\left(\frac{ks-k(1-s)}{s}\right)^{\frac{1}{s-1}},
	\end{eqnarray*}	a constant. This implies that $ m(t)$  is a linear function in $t$.
	Hence   $X$ has generalized power function distribution specified by (\ref{18}) (Elbatal et al. 2012).
\end{proof}
\section{Concluding remarks}
In this article, we study past entropy using a generating function approach.  Further, we established the relationships between past entropy-generating functions and reversed hazard rates.  Some characterization results are obtained. One can consider developing a goodness of fit test for power distribution using the characterization result obtained in this paper. Different extensions/modifications of measures of entropy and extropy are studied in the literature. We can study these measures using generating function approach.

\end{document}